\documentclass[12pt, final]{article}

\usepackage{amsmath, amsthm, amssymb, amsfonts}

\usepackage[letterpaper,margin=1.25in]{geometry} 
\usepackage{setspace}
\usepackage[dvipsnames]{xcolor}
\usepackage{tikz}
\usepackage{subfiles}
\usepackage{ifthen}

\theoremstyle{plain}
\newtheorem{theorem}{Theorem}

\newtheorem{proposition}{Proposition}

\theoremstyle{definition}

\usepackage{natbib}
\usepackage[
    colorlinks=true,
    linkcolor=blue,
    citecolor=blue,
    urlcolor=blue,
    filecolor=blue,
    hyperfootnotes=false,
    hyperindex,
    breaklinks
]{hyperref}

\usepackage{newtxtext} 
\usepackage{newtxmath}
\usepackage{microtype}
\microtypesetup{expansion=false} 

\usepackage[normalem]{ulem}

\title{Pair rationality and top trading cycles on single-peaked and single-dipped domains%
  \thanks{The authors used ChatGPT to assist with editing. The authors take full responsibility for the content of the paper, including its results, proofs, and references.}}
\author{
\"Ozg\"un Ekici\thanks{Department of Economics, \"Ozye\u{g}in University, Istanbul, T\"urkiye. Email: \texttt{ozgun.ekici@ozyegin.edu.tr}}
\and
M. Bumin Yenmez\thanks{Department of Economics, Washington University in St. Louis, USA; Durham University, UK; and \"Ozye\u{g}in University, T\"urkiye. Email: \texttt{bumin@wustl.edu}}
}
\date{}

\begin{document}
\maketitle

\begin{abstract}
In a recent study, \citet{EkiciYenmez2026} characterize top trading cycles (TTC) on the unrestricted domain by strategy-proofness, individual rationality, and pair rationality, uncovering a correspondence between the axiomatic foundations of TTC and deferred acceptance, matching theory's two canonical rules. We show that their characterization fails on the single-peaked domain, but it holds on the single-dipped domain even without strategy-proofness. They also show that pair rationality can be replaced in their characterization by two weaker conditions concerning agents' first and second choices. In contrast, this replacement fails on the single-dipped domain, even when strategy-proofness is imposed.
\end{abstract}

\begingroup\small
\endgroup

\section{Introduction}

In a classic result, \citet{Alcalde1994ET} show that, in the marriage problem, deferred acceptance (DA) is the unique rule satisfying individual rationality (IR), pair rationality (PR), and strategy-proofness (SP) for the proposing side. In a recent study, \citet{EkiciYenmez2026} establish a striking parallel in the object reallocation problem: top trading cycles (TTC) is the unique rule satisfying the same three axioms, with SP now required of all agents. Together, these characterizations reveal a common axiomatic foundation for DA and TTC, matching theory's two canonical rules.\footnote{DA and TTC were introduced by \citet{Gale1962AMM} and \citet{Shapley1974JME}, respectively. In the marriage problem, IR and PR together correspond to the notion of stability.}

We study whether this characterization survives under single-peakedness and single-dippedness. Both describe preferences over objects ordered by a common characteristic, such as size or location: preferences worsen away from an agent's peak under single-peakedness and improve away from her dip under single-dippedness.

We show that the characterization fails under single-peakedness (\autoref{prop:peaked}) but holds under single-dippedness even without SP (\autoref{thm:dipped}): IR and PR uniquely determine the TTC allocation at every profile. In contrast to the unrestricted-domain result of \citet{EkiciYenmez2026}, replacing PR with top-top rationality and top2-top rationality destroys uniqueness on the single-dipped domain, even when SP is imposed (\autoref{prop:weak}).

Our paper is most closely related to efficiency-based characterizations of TTC on restricted domains. \citet{Ma1994IJGT} characterizes TTC by IR, Pareto efficiency, and SP on the unrestricted domain. \citet{Bade2019} shows that this characterization fails on the single-peaked domain; \citet{Tamura2023} shows that it survives on the single-dipped domain. \citet{Ekici2024TE} strengthens Ma's result by replacing Pareto efficiency with pair efficiency; \citet{HuZhang2024} show that this stronger characterization also survives on the single-dipped domain. Notice that these efficiency notions rule out welfare-improving reallocations of assigned objects. By contrast, our analysis relies on PR, which rules out improvements through exchanges of endowments and is logically independent of Pareto efficiency and pair efficiency.

\section{Model}

Let $N=\{1,\ldots,n\}$ be a finite set of agents. Let $O=\{o_1,\ldots,o_n\}$ be a set of indivisible objects. Each agent $i$ initially owns object $o_i$ (i.e., $o_i$ is $i$'s \emph{endowment}).

Let $\mathcal P$ be the set of all strict preference orders over $O$. Agent $i$'s order is $P_i$: $a\,P_i\,b$ means that she prefers $a$ to $b$. Its weak counterpart is $R_i$, so $a\,R_i\,b$ means $a\,P_i\,b$ or $a=b$. A preference profile is $P=(P_i)_{i\in N}$. In explicit rankings, $\succ$ denotes strict preference.

An allocation is a bijection $\mu:N\to O$, where $\mu(i)$ is agent $i$'s assignment. Given a preference domain $\mathcal D\subseteq\mathcal P$, a rule $\phi$ selects an allocation at each $P\in\mathcal D^n$; $\phi_i(P)$ denotes agent $i$'s assignment at $P$.

We fix a common ordering of objects, $o_1<\cdots<o_n$. Preferences are \emph{single-peaked} if they worsen away from the most preferred object along either side of the ordering, and \emph{single-dipped} if they improve away from the least preferred object along either side of the ordering. Observe that if some objects are removed, preferences over the remaining objects are still single-peaked or single-dipped.

Fix a preference domain $\mathcal D\subseteq\mathcal P$. Our analysis centers on three axioms for rules on $\mathcal D^n$:

A rule $\phi$ satisfies \emph{strategy-proofness (SP)} if truthful reporting is always optimal: $\phi_i(P)\,R_i\,\phi_i(\widetilde P_i,P_{-i})$ for every $P\in\mathcal D^n$, $i\in N$, and $\widetilde P_i\in\mathcal D$, where $P_{-i}$ denotes the other agents' preferences.

An allocation $\mu$ satisfies \emph{individual rationality (IR)} at $P$ if each agent weakly prefers her assignment to her endowment: $\mu(i)\,R_i\,o_i$ for every $i \in N$. A rule $\phi$ satisfies IR if $\phi(P)$ satisfies IR at $P$ for every $P\in\mathcal D^n$.

An allocation $\mu$ satisfies \emph{pair rationality (PR)} at $P$ if no pair can improve upon their assignments by exchanging their endowments: there are no distinct agents $i$ and $j$ such that $o_j\,P_i\,\mu(i)$ and $o_i\,R_j\,\mu(j)$. A rule $\phi$ satisfies PR if $\phi(P)$ satisfies PR at $P$ for every $P\in\mathcal D^n$. Thus, PR rules out exchanges that strictly benefit one participant without harming the other.

The prominent rule in this problem is \emph{top trading cycles} (TTC), which we denote by $\varphi^{\mathrm{ttc}}$. TTC reallocates objects through successive rounds of trading. In the first round, each agent points to her favorite object, and each object points to its owner. This creates a directed graph containing at least one cycle. We write $(i_1,\ldots,i_m)$ for a cycle in which each agent points to the next agent's endowment, with $i_1$ being the next agent for $i_m$. An allocation \emph{executes} a cycle if each agent in it receives the object to which she points. TTC executes all cycles and removes their agents and objects. In each subsequent round, the remaining agents point to their favorite remaining objects, and we execute the resulting cycles in the same way. The procedure ends when every agent has received an object.

\section{Results}

\citet{EkiciYenmez2026} show that on the unrestricted domain (when $\mathcal D = \mathcal P$) TTC uniquely satisfies SP, IR, and PR. We first show that this unrestricted-domain characterization fails under single-peakedness.\footnote{\citet{Bade2019}'s crawler rule shows that SP, IR, and Pareto efficiency do not characterize TTC on the single-peaked domain. It cannot establish \autoref{prop:peaked} because it violates PR.}

\begin{proposition}\label{prop:peaked}
For every $n\geq3$, there is a rule on the single-peaked domain distinct from TTC that satisfies SP, IR, and PR.
\end{proposition}

\begin{proof}
We first suppose that $n=3$ and describe a rule $\phi \neq \varphi^{\mathrm{ttc}}$ that satisfies SP, IR, and PR.

At each profile $P$, let $\phi(P)=\varphi^{\mathrm{ttc}}(P)$ unless the initial graph contains a three-agent cycle. If there is such a cycle, it is $(1,2,3)$ or $(1,3,2)$. If the cycle is $(1,2,3)$, single-peakedness forces agent 2 to rank $o_3\succ o_2\succ o_1$ and agent 3 to rank $o_1\succ o_2\succ o_3$. If the cycle is $(1,3,2)$, single-peakedness forces agent 1 to rank $o_3\succ o_2\succ o_1$ and agent 2 to rank $o_1\succ o_2\succ o_3$. Notice that in each case at least one agent ranks her endowment second and at least one agent ranks her endowment last.

We complete the definition of $\phi$ as follows. If two agents rank their endowments last, write the cycle as $(i,j,k)$ with $k$ ranking hers second. Then agents $i$ and $j$ exchange their endowments and $k$ keeps hers. If two agents rank their endowments second, write the cycle as $(i,j,k)$ with $k$ ranking hers last. Then agent $i$ keeps her endowment and $j$ and $k$ exchange their endowments. Crucially, observe that in a three-agent cycle an agent receives her first choice if and only if the next two agents rank their endowments last and second, respectively.

\medskip

\noindent\textbf{Strategy-proofness.}
Fix agent $i$ and the others' reports and suppose $i$ is truthful. We show that $i$ cannot gain from manipulation.

If the others form a cycle without $i$, $\phi$ selects the TTC allocation for every report by $i$; then, since TTC satisfies SP, $i$ cannot gain from manipulation. Otherwise, there is one cycle and it contains $i$. If this cycle contains fewer than three agents, $\phi$ selects the TTC allocation and $i$ receives her first choice, so $i$ cannot gain from manipulation.

Therefore, suppose that the three-agent cycle $(i,j,k)$ forms. If $j$ and $k$ respectively rank their endowments last and second, $i$ receives her first choice, so she cannot gain from manipulation. Otherwise, $i$ receives her second choice, and a manipulation is successful only if it gives $i$ her first choice. But this is impossible: When $i$ misreports her first choice, this yields a singleton or two-agent cycle that contains $i$. Then $\phi$ selects the TTC allocation, so $i$ receives her reported first choice, which is not her true first choice. When $i$ switches her second and third choices, the cycle $(i,j,k)$ is preserved. Since the rankings of $j$ and $k$ remain unchanged, $i$ still cannot receive her true first choice. Therefore, $\phi$ satisfies SP.

\medskip
\noindent\textbf{Individual rationality.}
IR holds whenever $\phi$ selects the TTC allocation. Otherwise, every agent receives her first or second choice, while her endowment is ranked second or last. Therefore, $\phi$ satisfies IR.

\medskip
\noindent\textbf{Pair rationality.}
PR holds whenever $\phi$ selects the TTC allocation. Otherwise, write the cycle as $(i,j,k)$, where $i$ and $j$ exchange endowments and $k$ keeps hers. By construction, $k$'s ranking is $o_i\succ o_k\succ o_j$. Agent $i$ receives her first choice, $o_j$. Thus, an endowment exchange between $i$ and $k$ would make $i$ worse off, while one between $j$ and $k$ would make $k$ worse off. Finally, $i$ and $j$ already receive each other's endowments, so neither can improve through that exchange. Therefore, $\phi$ satisfies PR.

\medskip
\noindent\textbf{Extension to more than three agents.}
For $n>3$, we define the rule $\theta \neq \varphi^{\mathrm{ttc}}$ as follows: At profile $P$, if every agent $i\notin\{1,2,3\}$ ranks $o_i$ first, they are all assigned their endowments, and the assignments of agents $1,2,3$ are determined by the rule $\phi$, described above, using their preferences restricted to $\{o_1,o_2,o_3\}$. Observe that these restricted preferences remain single-peaked. At every other profile, let $\theta(P)=\varphi^{\mathrm{ttc}}(P)$. The rule $\theta$ satisfies SP, IR, and PR. The verification is straightforward and left to the reader.
\end{proof}

In contrast, the unrestricted-domain characterization survives and can even be strengthened under single-dippedness.

\begin{theorem}\label{thm:dipped}
On the single-dipped domain, TTC is the unique rule satisfying IR and PR.
\end{theorem}

\begin{proof}
TTC satisfies both axioms. Conversely, fix a single-dipped profile $P$ and an allocation $\mu$ satisfying IR and PR. We show that $\mu$ executes every TTC cycle.

Observe that, under single-dippedness, every agent's first choice is $o_1$ or $o_n$. These are the endowments of agents whose indices are smallest and largest. Thus, the first round of TTC produces either the cycle $(1,n)$, or one or both of the singleton cycles $(1)$ and $(n)$. By IR and PR, $\mu$ executes the resulting cycle or cycles.

After removing these agents and their endowments, $\mu$ assigns the remaining agents only remaining objects, and IR and PR still hold. Each remaining agent's favorite remaining object is the endowment of a remaining agent whose index is smallest or largest. Thus, repeating the same argument, $\mu$ executes the resulting cycle or cycles. Iterating, we obtain $\mu=\varphi^{\mathrm{ttc}}(P)$.
\end{proof}

\citet{EkiciYenmez2026} strengthen their unrestricted-domain characterization by replacing PR with two weaker axioms concerning a pair of agents' first and second choices. We next show that this replacement fails on the single-dipped domain.

An allocation $\mu$ satisfies \emph{top-top rationality (TTR)} at $P$ if, whenever distinct agents $i$ and $j$ rank each other's endowments first, $\mu(i)=o_j$ and $\mu(j)=o_i$. An allocation $\mu$ satisfies \emph{top2-top rationality (T2TR)} at $P$ if, whenever distinct agents $i$ and $j$ rank $o_j$ second and $o_i$ first, respectively, either $\mu(i)$ is agent $i$'s first choice or both $\mu(i)=o_j$ and $\mu(j)=o_i$.
A rule $\phi$ satisfies TTR (or T2TR) on $\mathcal D^n$ if, at every $P\in\mathcal D^n$, its chosen allocation $\phi(P)$ satisfies the corresponding axiom at $P$.

\begin{proposition}\label{prop:weak}
For every $n\geq4$, there is a rule on the single-dipped domain distinct from TTC that satisfies SP, IR, TTR, and T2TR.
\end{proposition}

\begin{proof}
Define the rule $\eta \neq \varphi^{\mathrm{ttc}}$ as follows. At each profile $P$, if agents $1$ and $n$ rank each other's endowments first, they exchange their endowments and every interior agent $i\notin\{1,n\}$ keeps hers. Otherwise, let $\eta(P)=\varphi^{\mathrm{ttc}}(P)$. The rule $\eta$ satisfies SP, IR, TTR, and T2TR. The verification is straightforward and left to the reader.
\end{proof}

\clearpage
\begingroup
\singlespacing
\setlength{\bibsep}{0pt}
\bibliography{matchingBib}
\endgroup
\end{document}